\documentclass[twocolumn]{revtex4-1}
\usepackage{amsmath}
\usepackage{amsthm}
\usepackage{latexsym}
\usepackage{amsfonts}
\usepackage{amssymb}
\usepackage{bbm,dsfont}
\usepackage{color}
\usepackage{graphicx}
\usepackage{enumerate}
\usepackage{subfigure}

\usepackage[normalem]{ulem}
\usepackage{mathrsfs}
\usepackage{mathtools}

\newtheorem{proposition}{Proposition}
\newtheorem{theorem}{Theorem}

\theoremstyle{definition}

\newtheorem{example}{Example}

\definecolor{darkgreen}{rgb}{0,0.4,0.3}

\newcommand{\real}{\mathbb R} 

\newcommand{\half}{\frac{1}{2}} 
\newcommand{\mo}[1]{\left| #1 \right|} 
\newcommand{\ca}[1]{\mathcal{#1}} 

\newcommand{\abs}{\mo} 

\newcommand{\ip}[2]{\left\langle\,#1\,|\,#2\,\right\rangle} 
\newcommand{\kb}[2]{|#1\rangle\langle#2|} 
\newcommand{\no}[1]{\left\|#1\right\|} 
\newcommand{\tr}[1]{{\rm tr}\left[#1\right]} 
\newcommand{\id}{\mathbbm{1}} 

\newcommand{\va}{\vec{a}} 
\newcommand{\vb}{\vec{b}} 
\newcommand{\vsigma}{\vec{\sigma}}

\newcommand{\M}{\mathsf{M}}
\newcommand{\N}{\mathsf{N}}
\newcommand{\G}{\mathsf{G}}
\renewcommand{\P}{\mathsf{P}}
\newcommand{\Q}{\mathsf{Q}}

\newcommand{\en}{\mathcal{E}} 

\newcommand{\Pqrac}{P_{{\rm qrac}}}
\newcommand{\Prac}{P_{{\rm rac}}}
\newcommand{\Pqracopt}{\bar{P}_{{\rm qrac}}}

\newcommand{\phii}{\varphi}

\begin{document}

\title[]{Quantum random access codes and incompatibility of measurements}

\author{Claudio Carmeli}
\email{claudio.carmeli@gmail.com}
\affiliation{DIME, Universit\`a di Genova, Via Magliotto 2, I-17100 Savona, Italy}

\author{Teiko Heinosaari}
\email{teiko.heinosaari@utu.fi}
\affiliation{QTF Centre of Excellence, Turku Centre for Quantum Physics, Department of Physics and Astronomy, University of Turku, FI-20014 Turku, Finland}

\author{Alessandro Toigo}
\email{alessandro.toigo@polimi.it}
\affiliation{Dipartimento di Matematica, Politecnico di Milano, Piazza Leonardo da Vinci 32, I-20133 Milano, Italy}
\affiliation{I.N.F.N., Sezione di Milano, Via Celoria 16, I-20133 Milano, Italy}

\begin{abstract}
We prove that quantum random access code (QRAC) performs better than its classical counterpart only when incompatible quantum measurements are used in the decoding task. As a consequence, evaluating the average success probability for QRAC provides a semi-device independent test for the detection of quantum incompatibility. We further demonstrate that any incompatible pair of projective measurements gives an advantage over all classical strategies. Finally, we establish a connection between the maximal average success probability for QRAC and earlier quantities introduced to assess incompatibility.
\end{abstract}

\maketitle

\section{Introduction}\label{sec:intro}

Random access codes (RACs) are an important type of communication tasks where a sender (Alice) encodes a string of bits into a single bit and a receiver (Bob) aims to recover some randomly chosen subset of the data.
It has been shown that the probability of Bob to access the randomly selected part of the information can be increased if Alice sends qubits instead of bits \cite{Wi83,AmNaTSVa02,Ga02,HaIwNiRaYa06}, or more generally, qudits instead of dits \cite{TaHaMaBo15}. 
These scenarios are called quantum random access codes (QRACs) and they have been investigated from various different angles in several recent works \cite{AgBoMiPa18,CzSaTaPa18,FaKa19,MiPa19}.

In this work we are interested in the question what quantum resources are needed in order to achieve an average success probability in QRAC over the classical bound.
We show that incompatibility of measurements is a necessary ingredient to enable QRAC protocol to go over the classical bound in some given dimension. 
We further demonstrate that it is not generally sufficient, but we identify some important classes of measurements where it is.
In particular, we show that a pair of sharp measurements enable to go over the classical bound if and only if they are incompatible.

We present a generalization of QRAC and we demonstrate that this can be used to detect incompatibility also in the cases that do not fall into the realm of the usual QRAC. 
Namely, we consider a prepare-measure scenario of the same type as QRAC, but where the number of outcomes of measurements can be different than the Hilbert space dimension of the communication medium.
We derive an upper bound for the success probability that is satisfied by all compatible pairs of measurements but that can be exceeded by some incompatible measurements.
This scenario thus gives a semi-device-independent detection of quantum incompatibility.

\section{QRAC}

Quantum random access code (QRAC) is a specific kind of communication task. 
In the $(n,d)$-QRAC scheme Alice receives $n$ dits, $\vec{x}=(x_1,\ldots,x_n)$.
She can send one qudit to Bob. 
In addition, Bob receives a number $j\in\{1,\ldots,n\}$ and
his task is to guess the corresponding dit $x_j$. 
He does this by performing a measurement, depending on $j$, thereby obtaining an outcome $y$.
The enconding-decoding is successful if $y=x_j$.
The total success probability is usually taken to be either the worst case success probability or the average success probability, and when calculating these numbers it is assumed that the inputs for both Alice and Bob are uniformly distributed.

The strategy of Alice and Bob consist of $d^n$ quantum states for encoding and $n$ $d$-outcome measurements for decoding, all defined for a $d$-level quantum system.
We denote by $\en$ the encoding map and $\M_1,\dots,\M_n$ the measurements.
Hence, $\en(\vec{x})$ is a quantum state for each $\vec{x}$, and $\M_1,\dots,\M_n$ are $d$-outcome positive operator valued measures (POVMs).
The average success probability is then given as
\begin{align}\label{eq:average}
\frac{1}{nd^n} \sum_{\vec{x}} \tr{\en(\vec{x}) (\M_1(x_1) + \cdots + \M_n(x_n))} \, .
\end{align}
We denote by $\Pqrac^{n,d}$ the best achievable average success probability in $(n,d)$-QRAC.
In the case $n=2$, it is known from \cite{TaHaMaBo15} that $\Pqrac^{2,d} = \half (1+1/\sqrt{d})$.

The most interesting fact about QRAC is revealed when it is compared to its classical counterpart.
The rules in the classical random access code (RAC) are otherwise the same as described earlier, but now Alice is allowed to send a dit to Bob instead of a qudit.
We denote by $\Prac^{n,d}$ the optimal average success probability in $(n,d)$-RAC.
It is known \cite{TaHaMaBo15,AmKrAs15} that $\Prac^{2,d} = \half (1+\tfrac{1}{d})$.
In particular, as $\Pqrac^{2,d} > \Prac^{2,d}$, a suitably chosen quantum strategy can be better than any classical strategy.

Let us consider a strategy for $(n,d)$-QRAC where decoding is done by measurements $\M_1,\ldots,\M_n$.
As noted in \cite{FaKa19}, from \eqref{eq:average} we see that for a fixed collection of measurements, the average success probability is maximized when $\en(\vec{x})$ is an eigenstate corresponding to the maximal eigenvalue of the operator $\M_1(x_1)+\cdots +\M_n(x_n)$.
We write $\Pqracopt(\M_1,\ldots,\M_n)$ for the average success probability of a given collection of measurements when the states are optimized, and hence conclude that
\begin{align}\label{eq:pqracopt}
\Pqracopt(\M_1,\ldots,\M_n) = \frac{1}{nd^n} \sum_{\vec{x}} \no{\M_1(x_1) + \cdots + \M_n(x_n)} \, ,
\end{align}
where $\no{\cdot}$ is the operator norm.
We say that quantum measurements $\M_1,\ldots,\M_n$ are useful for $(n,d)$-QRAC if they enable to go over the classical bound, i.e., $\Pqracopt(\M_1,\ldots,\M_n)>\Prac^{n,d}$.

From now on, we concentrate on $(2,d)$-QRAC.
From the previous discussion, it follows that two measurements $\M_1$ and $\M_2$ are useful for $(2,d)$-QRAC if and only if
\begin{align}\label{eq:useful}
\sum_{x,y} \no{\M_1(x) + \M_2(y)} >  d(d+1) \, .
\end{align}
From \eqref{eq:useful}, we can verify the rather obvious but important fact that not all pairs of quantum measurements are useful for QRAC.
For instance, let us consider the choice $\M_1=\M_2=\M$. 
We get
\begin{equation}\label{eq:long}
\begin{aligned}
& \sum_{x,y}  \no{\M(x) + \M(y)} = 2\sum_x \no{\M(x)} \\
& \quad + \sum_{\substack{x,y \\ x\neq y}}  \no{\M(x) + \M(y)} \leq \sum_x  2 \no{\M(x)} + \sum_{\substack{x,y \\ x\neq y}}  \no{\id} \\
& \ \leq 2d + d(d-1) = d(d+1) \, , 
\end{aligned}
\end{equation}
which means that $\Pqracopt(\M,\M)\leq \Prac^{2,d}$.
Let us note that even this foolish quantum strategy (i.e., choosing $\M_1=\M_2$) reaches the optimal classical success probability under the condition that $\no{\M(x)}=1$ for all $x$.
This is the case when $\M$ is \emph{sharp}, i.e., $\M(x)=\kb{\psi_x}{\psi_x}$ for an orthonormal basis $\{\psi_x\}_{x=1}^d$.

\section{Necessity of incompatibility}

We recall that two measurements $\M_1$ and $\M_2$ are called \emph{compatible} if there is a third measurement $\G$, defined on their product outcome set, such that
\begin{align*}
\sum_y \G(x,y)=\M_1(x) \, , \quad \sum_x \G(x,y)=\M_2(y) \, .
\end{align*}
Otherwise $\M_1$ and $\M_2$ are \emph{incompatible}.
Incompatibility is a genuine quantum property that makes quantum measurements different from classical measurements \cite{HeMiZi16}.
Therefore, one can expect that incompatibility is related to the usefulness of quantum measurements for QRAC.
In fact, we have the following result.

\begin{theorem}\label{prop:compatible}
A compatible pair of $d$-outcome measurements is not useful for $(2,d)-QRAC$, i.e., 
compatible measurements $\M_1$ and $\M_2$ satisfy
\begin{align}
\Pqracopt(\M_1,\M_2) \leq \Prac^{2,d} \, .
\end{align}
\end{theorem}

\begin{proof}
Let us assume that $\M_1$ and $\M_2$ are compatible $d$-outcome measurements, hence having a joint measurement $\G$ with $d^2$ outcomes.
We have
\begin{align*}
&2d^2 \cdot \Pqracopt(\M_1,\M_2) = \\
&= \sum_{x,y}\bigg\|\sum_a  \G (x,a) + \sum_b \G(b,y)\bigg\| \\
& = \sum_{x,y} \Bigg\|\G(x,y)+ \sum_{\substack{a\\ a \neq y}}  \G (x,a) + \sum_{b} \G(b,y)\Bigg\| \\
& \leq \sum_{x,y} \no{\G(x,y)}+ \sum_{x,y}\Bigg\|\sum_{\substack{a\\ a \neq y}}  \G (x,a) + \sum_{b} \G(b,y)\Bigg\| \, .
\end{align*}
We estimate these two terms separately. 
For the first summand we have
\begin{align*}
& \sum_{x,y} \no{ \G(x,y)} \leq \sum_{x,y} \tr{\G(x,y)} = \tr{\id} = d\,.
\end{align*}
For the latter summand, we observe that
\begin{equation*}
0\leq \sum_{\substack{a\\ a \neq y}}  \G (x,a) + \sum_{b} \G(b,y) \leq \sum_{a,b} \G(a,b) =\id
\end{equation*}
and therefore
\begin{align*}
\sum_{x,y} \Bigg\|\sum_{\substack{a\\ a \neq y}}  \G (x,a) + \sum_{b} \G(b,y)\Bigg\| \leq d^2 \, . 
\end{align*}
Putting these together we get
\begin{align*}
\Pqracopt(\M_1,\M_2) \leq \frac{1}{2d^2} (d+d^2)= \frac12 \left(1 + \frac{1}{d} \right) \, .
\end{align*}
\end{proof}

Due to our earlier observation that $\Pqrac(\M,\M)=\Prac^{2,d}$ for certain $d$-outcome measurements $\M$, we conclude that the upper bound given in Thm.~\ref{prop:compatible} is tight.
The obvious question is: does the converse of Thm.~\ref{prop:compatible} hold, namely, are all pairs of incompatible measurements useful for QRAC?
In the following we see that this is not the case, but it holds under some additional assumptions.

\section{Sharp measurements}

A sharp $d$-outcome measurement $\M$ is defined as $\M(x)=\kb{\psi_x}{\psi_x}$ for an orthonormal basis $\{\psi_x\}_{x=1}^d$.
It is known that a sharp measurement is compatible with another measurement if and only if they are commuting \cite{HeReSt08}.
It follows that two sharp $d$-outcome measurements $\M_1$ and $\M_2$ are compatible if and only if there is a permutation $\sigma$ of the set $\{1,\ldots,d\}$ such that $\M_2(y) = \M_1(\sigma(y))$ for all $y$.
In this case $\M_1$ and $\M_2$ are, essentially, describing the same measurement but they differ by relabeling of the outcomes. 
It can be verified as in \eqref{eq:long} that $\Pqracopt(\M_1,\M_2)=\Pqracopt(\M_1,\M_1)=\Prac^{2,d}$. 

Two sharp $d$-outcome measurements that are not related by a permutation of outcomes are incompatible and hence satisfy the necessary criterion for being useful for QRAC. 
The immediate question is: is an arbitrary pair of incompatible sharp measurements useful for QRAC? 
This has been conjectured to be so and called Homer's conjecture \cite{AmLeMaOz08}.
In the following proposition we prove that the conjecture is, indeed, true in $(2,d)$-QRAC.

\begin{proposition}
Let $\M_1$ and $\M_2$ be two sharp $d$-outcome measurements. 
Then
\begin{align}
\Pqracopt(\M_1,\M_2) \geq \Prac^{2,d} \, ,
\end{align}
with equality attained if and only if $\M_1$ and $\M_2$ are compatible (which is the case if they are the same up to permutation of outcomes). 
\end{proposition}

\begin{proof}
Fix two orthonormal bases $\{\phii_x\}_{x=1}^d$ and $\{\psi_y\}_{y=1}^d$ such that $\M_1(x) = \kb{\phii_x}{\phii_x}$ and $\M_2(y) = \kb{\psi_y}{\psi_y}$ for all $x,y=1,\ldots,d$. Moreover, let $\ip{\phii_x}{\psi_y} = \kappa(x,y) {\rm e}^{{\rm i}\theta(x,y)}$ with $0\leq \kappa(x,y) = \abs{\ip{\phii_x}{\psi_y}} \leq 1$ and $0\leq \theta(x,y)<2\pi$. 
The nonzero eigenvalues of the positive operator $\M_1(x)+\M_2(y)$ are $1\pm\kappa(x,y)$ with respective eigenvectors $\psi_y\pm{\rm e}^{{\rm i}\theta(x,y)}\phii_x$. Therefore, $\no{\M_1(x)+\M_2(y)} = 1+\kappa(x,y)$ and
\begin{align*}
& \sum_{x,y} \no{\M_1(x)+\M_2(y)} = d^2 + \sum_{x,y} \kappa(x,y) \\
&\qquad\qquad\qquad \geq d^2 + \sum_{x,y} \kappa(x,y)^2 = d^2+d\,.
\end{align*}
In the latter expression, all the equalities are attained if and only if $\kappa(x,y)\in\{0,1\}$ for all $x,y$, which amounts to the two bases $\{\phii_x\}_{x=1}^d$ and $\{\psi_y\}_{y=1}^d$ being equal up to permutation and multiplication by phase factors. The claim then follows by comparing with \eqref{eq:useful}.
\end{proof}

\section{Dichotomic qubit measurements}

A dichotomic qubit measurement is parametrized by a vector $(\alpha,\va)\in\real^4$ as follows:
\begin{align*}
\M_{\alpha,\va}(\pm 1) =\left[(1\pm\alpha) \id \pm \va \cdot \vsigma\right] / 2 \, .
\end{align*}
The parameters must satisfy $\no{\va}\leq 1$ and $\mo{\alpha}\leq 1-\no{\va}$ in order for $\M_{\alpha,\va}$ to be a valid POVM.
A measurement $\M_{\alpha,\va}$ is called \emph{unbiased} if $\alpha=0$; otherwise \emph{biased}.

For two qubit measurements $\M_{\alpha,\va}$ and $\M_{\beta,\vb}$ we obtain
\begin{equation}
\big\|\M_{\alpha,\va}(x) + \M_{\beta,\vb}(y)\big\| = 1 + \half \big\|x\va + y\vb\big\| + \half (\alpha x + \beta y)
\end{equation}
for all $x,y = \pm 1$, and hence
\begin{align}\label{eq:qrac-qubit}
\Pqracopt(\M_{\alpha,\va},\M_{\beta,\vb})=\frac{1}{2} + \frac{1}{8} \left[ \big\|\va+\vb\big\| + \big\|\va-\vb\big\| \right] \, .
\end{align}
Interestingly, the bias parameters $\alpha$ and $\beta$ do not enter into this expression.

As proven in \cite{Busch86}, two unbiased qubit measurements $\M_{0,\va}$ and $\M_{0,\vb}$ are incompatible if and only if
\begin{align}\label{eq:paul}
\big\|\va+\vb\big\| + \big\|\va-\vb\big\| > 2 \, .
\end{align}
A comparison of this inequality to \eqref{eq:useful} and \eqref{eq:qrac-qubit} leads to the following conclusion.

\begin{proposition}\label{prop:unbiased}
Two unbiased qubit measurements $\M_1$ and $\M_2$ are incompatible if and only if they are useful for $(2,2)$-QRAC.
\end{proposition}

Two biased qubit observables can be incompatible but have $\Pqracopt(\M_1,\M_2) < \Prac^{2,2}$.
For instance, let us choose $\alpha=0,\beta=0.2$ and $\va=(0.7,0,0),\vb=(0,0.7,0)$. 
As proven in \cite{StReHe08,BuSc10,YuLiLiOh10}, the corresponding measurements are incompatible, but from \eqref{eq:qrac-qubit} we see that they are not useful for (2,2)-QRAC.

\section{Mutually unbiased measurements with noise}

We have previously learned that incompatibility is only a necessary condition for two measurements to be useful for QRAC, the sufficiency requiring some additional assumptions.
One could hope that a result analogous to Prop.~\ref{prop:unbiased} is valid also in higher dimensions. 
To study this question, we look noisy versions of two mutually unbiased (MU) measurements in dimension $d$. 
In $d=2$ these would correspond to $\M_{0,\va}$ and $\M_{0,\vb}$ with $\va\cdot\vb=0$.
It is known that the incompatibility of two MU measurements is more resilient to uniform noise when the dimension $d$ increases \cite{UoLuMoHe16,DeSkFrBr19,CaHeTo19}.
In the following we demonstrate that for $(2,d)$-QRAC the trend is opposite, i.e., in a higher dimension two MU measurements tolerate less noise for being useful for QRAC.

Fix $d\geq 2$ and let $\{\varphi_x\}_{x=1}^d$ and $\{\psi_y\}_{y=1}^d$ be mutually unbiased bases, i.e., $\mo{\ip{\varphi_x}{\psi_y}}=1/\sqrt{d}$ for all $x,y$.
We are considering noisy $d$-outcome measurements defined as
\begin{align}
\Q_\mu(x) & = \mu\,\kb{\varphi_x}{\varphi_x} + (1-\mu)\,\id/d \\
\P_\nu(y) & = \nu\,\kb{\psi_y}{\psi_y} + (1-\nu)\,\id/d
\end{align}
with $\mu,\nu\in[0,1]$ being the noise parameters.  
It follows from \cite[Supplementary Material]{CaHeTo19} that for any $x,y=1,\ldots,d$, we have
\begin{equation}
\begin{aligned}
\no{\Q_\mu(x) + \P_\nu(y)} = & \,\frac{2}{d} + \frac{d-2}{2d}(\mu+\nu) \\
& + \frac{1}{2}\sqrt{\mu^2+\nu^2-\frac{2(d-2)}{d}\mu\nu} \,.
\end{aligned}
\end{equation}
Therefore, $\Q_\mu$ and $\P_\nu$ are useful for $(2,d)$-QRAC if and only if
\begin{equation}\label{eq:f-ineq}
f_d(\mu,\nu) > d-1 \, ,
\end{equation}
where 
\begin{equation}\label{eq:f}
f_d(\mu,\nu)= (d-2)(\mu+\nu-\mu\nu) + \mu^2+\nu^2 \, .
\end{equation}
Assuming that $\mu+\nu > 1$ (the bound for trivial compatibility \cite{BuHeScSt13}), in \cite{CaHeTo19} the necessary and sufficient condition for  $\Q_\mu$ and $\P_\nu$ to be incompatible has been shown to be
\begin{equation}\label{eq:g-ineq}
g_d(\mu,\nu) < d-1 \, ,
\end{equation}
with
\begin{equation}\label{eq:g}
g_d(\mu,\nu) = 2(d-2)(\mu+\nu-\mu\nu)-d(\mu^2+\nu^2) +3 \, .
\end{equation}
One can readily verify that the conditions \eqref{eq:f-ineq} and \eqref{eq:g-ineq} are equivalent for $d=2$, which is consistent with Prop.~\ref{prop:unbiased}.
For all $d\geq 3$ the first condition is strictly tighter than the second one (see Appendix \ref{app:A} for details).
This means that a pair of noisy MU measurements can be incompatible without being useful for QRAC.
In fact, when the dimension $d$ increases, the region of points $(\mu,\nu)$ that correspond to incompatible pairs becomes larger whereas the region where they correspond to useful pairs for QRAC becomes smaller; see Fig. \ref{fig:regions}.
In the limit $d\to\infty$, the incomptibility condition becomes $\mu+\nu > 1$ while the QRAC condition becomes $\mu=\nu=1$. 

\begin{figure}
    \centering
    \subfigure[]
    {
        \includegraphics[width=3.3cm]{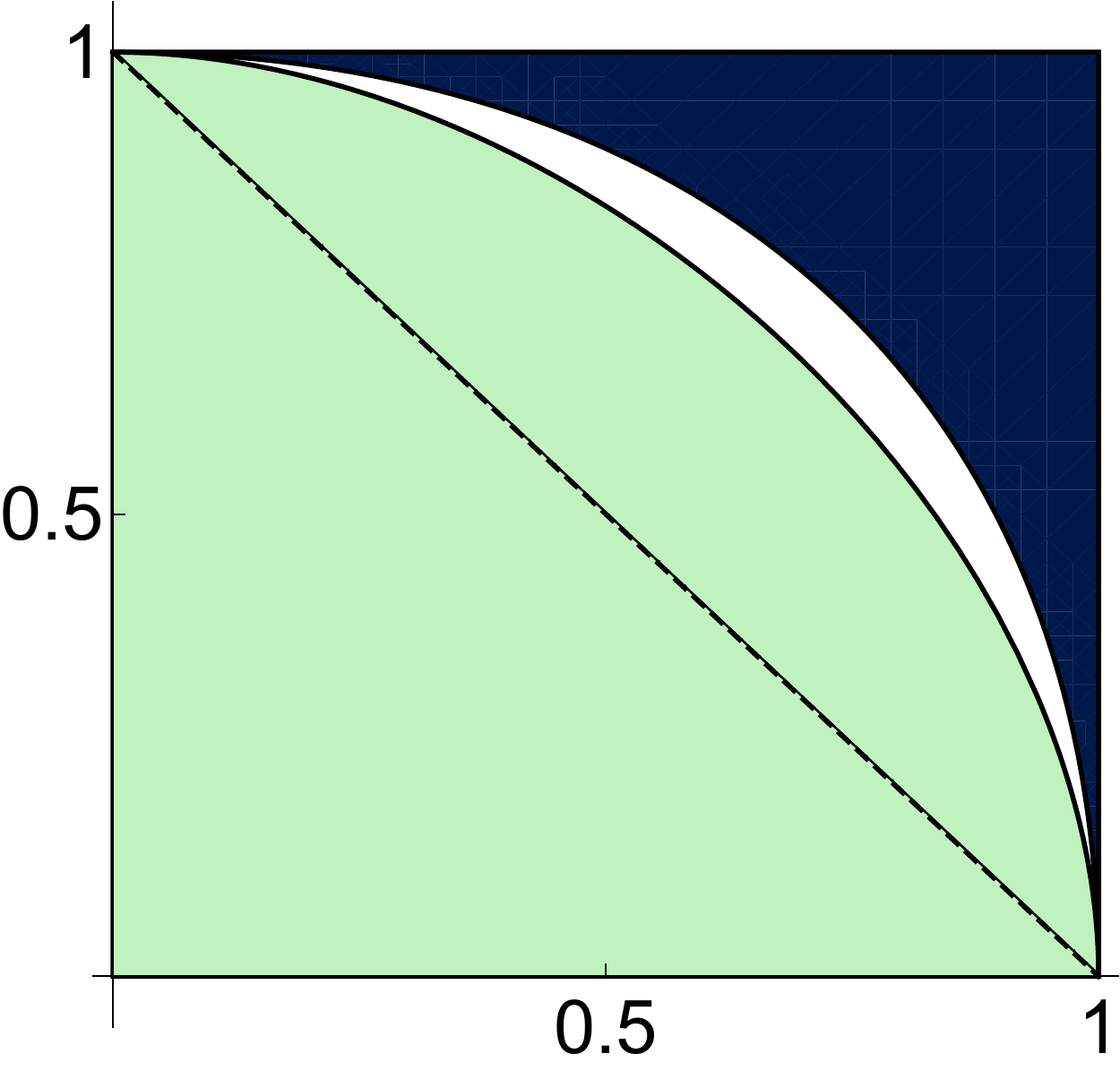} 
    }
    \subfigure[]
    {
        \includegraphics[width=3.3cm]{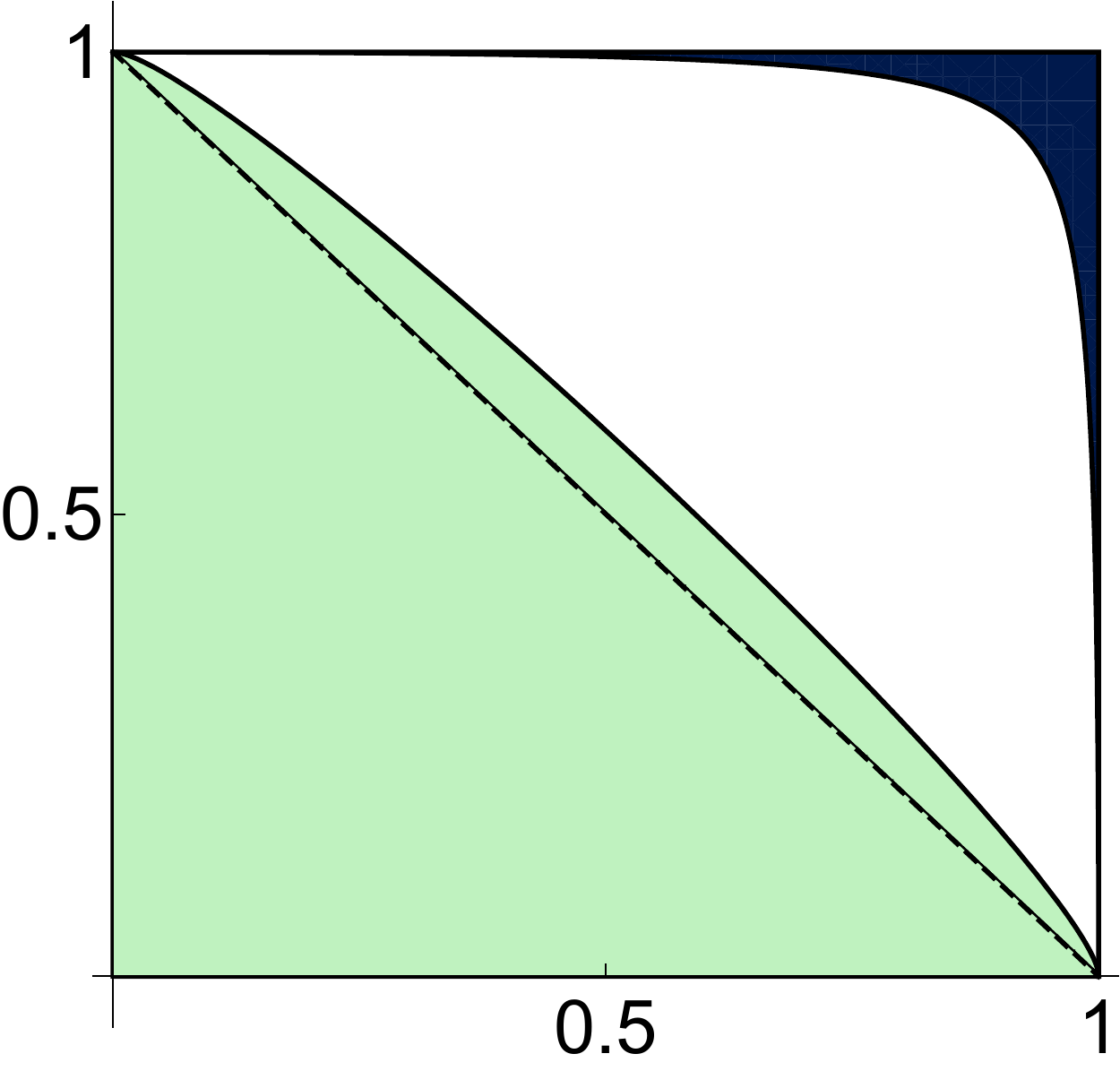}   
    }
        \caption{The regions of points $(\mu,\nu)$ that correspond to compatible pairs (light green) and pairs that are useful for $(2,d)$-QRAC (dark blue). In the intermediate white region, pairs are incompatible but useless for $(2,d)$-QRAC. (a) $d=3$ and (b) $d=100$.  }
    \label{fig:regions}
\end{figure}

\section{Detecting incompatibility}

An obvious implication from Thm.~\ref{prop:compatible} is that if we run a QRAC test and get an average success probability over the classical bound, then the measurements involved are incompatible. 
For the purpose of incompatibility detection, we can generalize the QRAC setting and in this way allow a wider range of applicability.

Let $\en$ be a labeled set of $n\times m $ states of a $d$-dimensional quantum system, and let $\M_1$ and $\M_2$ be measurements with $n$ and $m$ outcomes, respectively.
The setting is otherwise the same as in QRAC; Alice sends a state $\en(i,j)$ to Bob with the probability $1/nm$, Bob receives an input $k\in\{1,2\}$, and then he must try to guess the encoded message $x_k$ by making a measurement $\M_k$ on $\en(i,j)$.
The average success probability is 
\begin{equation}\label{eq:average_mn}
\frac{1}{2nm} \sum_{x=1}^n \sum_{y=1}^m\tr{\en(x,y) (\M_1 (x) + \M_2(y))} \, .
\end{equation}
We denote by $\bar{P}(\M_1,\M_2)$ the average success probability for fixed  measurements $\M_1$ and $\M_2$ when the encoding states are optimized.
Hence,
\begin{align}\label{eq:pqracopt_mn}
\bar{P}(\M_1,\M_2)= \frac{1}{2nm} \sum_{x=1}^n \sum_{y=1}^m \no{\M_1 (x) + \M_2(y)} \, .
\end{align}
The proof of the following proposition is essentially the same as that of Thm.~\ref{prop:compatible}.

\begin{proposition}\label{prop:compatible-nm}
If 
\begin{align}\label{eq:compatible-nm}
\bar{P}(\M_1,\M_2) >  \frac{1}{2} \left( 1 + \frac{d}{nm} \right) \, ,
\end{align}
then $\M_1$ and $\M_2$ are incompatible.
\end{proposition}

This result gives a semi-device-independent method to detect incompatibility; if for some collection of states the average success probability \eqref{eq:average_mn} is greater than the bound given in \eqref{eq:compatible-nm}, then we can conclude that $\M_1$ and $\M_2$ are incompatible even if we wouldn't know the used states.

Since always $\bar{P}(\M_1,\M_2)\leq 1$, the lower bound in \eqref{eq:compatible-nm} is pointless for $nm\leq d$.
The lower bound is useless also when the product $nm$ is large compared to $d$.
Namely, any two measurements $\M_1$ and $\M_2$ on a $d$-dimensional system satisfy 
$\sum_x \no{\M(x)} \leq d$
and therefore
\begin{align*}
\sum_{x=1}^n \sum_{y=1}^m \no{\M_1 (x) + \M_2(y)} & \leq \sum_{x=1}^n \sum_{y=1}^m \tr{\M_1 (x) + \M_2(y)} \\
& = \tr{m\,\id + n\,\id} =  d(n+m) \, .
\end{align*}
This implies that
\begin{align}\label{eq:Pbound1}
\bar{P}(\M_1,\M_2) \leq \frac{d(n+m)}{2nm} 
\end{align}
regardless whether $\M_1$ and $\M_2$ are compatible or incompatible.
Hence \eqref{eq:compatible-nm} is a nontrivial condition only when 
$d(n+m) > d+nm$.
For instance, for $d=2$ the cases when \eqref{eq:compatible-nm} can hold are when $n,m\in\{2,3,4\}$ and $nm\leq 9$.
In Appendix \ref{app:B} we present some examples outside of the usual QRAC where the incompatibility of pair of measurements is detected by the presented condition \eqref{eq:compatible-nm}.

An important property of the quantity $\bar{P}(\M_1,\M_2)$ is that it scales linearly if $\M_1$ and $\M_2$ are mixed with equal amount of noise. Namely, let $p_1$ and $p_2$ be probability distributions with $n$ and $m$ outcomes, respectively. Mixtures $t\M_1+(1-t)p_1\id$ and $t\M_2+(1-t)p_2\id$ describe noisy versions of $\M_1$ and $\M_2$, and we get
\begin{equation}\label{eq:p1p2}
\begin{aligned}
& \bar{P}(t\M_1+(1-t)p_1\id,\,t\M_2+(1-t)p_2\id) = \\
& \qquad = \frac{1}{2nm} \sum_{x=1}^n \sum_{y=1}^m \|t(\M_1 (x) + \M_2(y))\\
& \qquad\quad +(1-t)(p_1(x)+p_2(y))\id\| \\
& \qquad = t\,\bar{P}(\M_1,\M_2)+ (1-t)\,\frac{n+m}{2nm}\, .
\end{aligned}
\end{equation}
We observe that the specific form of $p_1$ and $p_2$ does not appear in the last expression. This scaling property of $\bar{P}(\M_1,\M_2)$ implies that, whenever \eqref{eq:compatible-nm} is satisfied, the joint measurability degree \cite{HeScToZi14}
\begin{equation}\label{eq:meas_degree}
\begin{aligned}
& \mathsf{j}(\M_1,\M_2) = \max\{t\geq 0 \mid \{t\M_i + (1-t)p_i\id\}_{i=1,2} \\
& \qquad\text{are compatible for some probabilities $p_1,p_2$}\}
\end{aligned}
\end{equation}
obeys the upper bound
\begin{equation}\label{eq:meas_degree_bound}
\mathsf{j}(\M_1,\M_2) \leq \frac{d+nm-(n+m)}{2nm\bar{P}(\M_1,\M_2)-(n+m)} \,.
\end{equation}
Further, one can also prove that the incompatibility robustness \cite{Haapasalo15,UoBuGuPe15}
\begin{equation}\label{eq:robustness}
\begin{aligned}
& \ca{R}(\M_1,\M_2) = \min\{t\geq 0\mid\{(\M_i + t\N_i)/(1+t)\}_{i=1,2} \\
& \quad \text{are compatible for some measurements $\N_1,\N_2$}\}
\end{aligned}
\end{equation}
is lower bounded by
\begin{equation}\label{eq:robustness_bound}
\ca{R}(\M_1,\M_2) \geq \frac{2nm\bar{P}(\M_1,\M_2)}{d+nm} -1\,.
\end{equation}
The detailed proofs of \eqref{eq:meas_degree_bound} and \eqref{eq:robustness_bound} are given in Appendix \ref{app:C}.
A different way to bound these kind of robustness quantities has been presented in \cite{FaKa19}.

\section{Discussion}

We have studied $(2,d)$-QRAC and proven that for two measurements to go over the classical bound, they must be incompatible. 
We have further shown that for certain types of measurements, every incompatibility pair gives quantum advantage with a suitable encoding. 
Finally, we have presented a generalized setting of QRAC that can be used for incompatibility detection even when the number of outcomes does not coincide with $d$.

Recently, an incompatibility detection method based on state discrimination task has been presented and studied \cite{CaHeTo19,SkSuCa19,UoKrShYuGu19,BuChZh19}. 
In the approach of \cite{CaHeTo19}, such a task has been connected to the notion of incompatibility witness \cite{Jencova18,BlNe18}.
The difference of (generalized) QRAC tests to incompatibility witnesses is that the latter is a universal method, i.e., every incompatible pair of measurements gives an advantage over compatible measurements in some discrimination task. But the merit of QRAC is that unlike an incompatibility witness, it is a semi-device-independent test of incompatiblity. 
An outstanding open problem is if all incompatible pairs can be detected in some semi-device-independent way, possibly with some variation of QRAC tests.

\section*{Acknowledgements}
T.H. acknowledges financial support from the Academy of Finland via the Centre of Excellence program (Grant No. 312058) as well as Grant No. 287750.

\appendix

\section{Noisy MU measurements}\label{app:A}

In this appendix, we prove some properties of the functions $f_d$ and $g_d$ that were defined in \eqref{eq:f} and \eqref{eq:g}.

\begin{proposition}
If $d\geq 2$, then the inequality $f_d(\mu,\nu) > d-1$ implies that $g_d(\mu,\nu) < d-1$ and $\mu+\nu>1$ for all $\mu,\nu\in[0,1]$.
For $d=2$, also the converse implication holds, whereas for every $d\geq 3$ it is violated by some $\mu,\nu\in[0,1]$.
\end{proposition}

\begin{proof}
The inequality $f_d(\mu,\nu) > d-1$ is equivalent to
\begin{equation}\label{eq:mu2+nu2}\tag{$\ast$}
\mu^2+\nu^2 > d-1+(d-2)(\mu\nu-\mu-\nu) \,.
\end{equation}
Inserting it into the definition \eqref{eq:g} of $g_d$, we obtain
$$
g_d(\mu,\nu) < (d^2-4)(\mu+\nu-\mu\nu) - d^2 + d  + 3 \,.
$$
Since
$$
\max_{\mu,\nu\in[0,1]} (\mu+\nu-\mu\nu) = 1 \,,
$$
it follows that $g_d(\mu,\nu) < d - 1$. Moreover, \eqref{eq:mu2+nu2} also implies
\begin{align*}
& \mu+\nu\geq \mu^2+\nu^2 > d-1+(d-2)(\mu\nu-\mu-\nu) \\
& \qquad \Rightarrow \quad \mu+\nu > 1 + \frac{d-2}{d-1}\,\mu\nu \geq 1 \,.
\end{align*}
For $d=2$, the conditions $f_d(\mu,\nu) > d-1$ and $g_d(\mu,\nu) < d-1$ are equivalent by direct inspection. If $d\geq 3$, then we have
$$
\begin{cases}
\mu+\nu > 1 \\
g_d(\mu,\nu) < d-1 \\
f_d(\mu,\nu) < d-1
\end{cases}
$$
for all $\mu=\nu$ in the open interval
$$
\left(\frac{\sqrt{d}+2}{2(\sqrt{d}+1)} \,,\, \frac{\sqrt{d}+1}{\sqrt{d}+2}\right) \,.
$$
\end{proof}

\begin{proposition}
If $\mu,\nu\in[0,1]$ and $d\geq 2$, then, for all $d'<d$,
\begin{itemize}
\item $f_d(\mu,\nu) > d-1$ implies $f_{d'}(\mu,\nu) > d'-1$;
\item $g_d(\mu,\nu) < d-1$ implies $g_{d'}(\mu,\nu) < d'-1$.
\end{itemize}
\end{proposition}

\begin{proof}
We have
$$
\frac{\partial}{\partial d}\left[f_d(\mu,\nu) - (d-1)\right] = \mu+\nu-\mu\nu-1 \leq 0
$$
for all $\mu,\nu\in[0,1]$, thus impliying that the function $d\mapsto f_d(\mu,\nu) - (d-1)$ is nonincreasing. Similarly,
\begin{align*}
\frac{\partial}{\partial d}\left[g_d(\mu,\nu) - (d-1)\right] & = 2(\mu+\nu-\mu\nu)-(\mu^2+\nu^2)-1 \\
& = -\left(\mu+\nu-1\right)^2 \leq 0 \,.
\end{align*}
\end{proof}

\section{Examples of incompatibility detection}\label{app:B}

In the following examples we demonstrate that Prop.~\ref{prop:compatible-nm} can be used to detect incompatibility in some cases when we are not in the usual $(2,d)$-QRAC setting with $d=n=m$.
Below, we denote $\M_k(\pm 1) =\half (\id \pm \sigma_k)$ for $k=1,2,3$.

\begin{example}($d=2$, $m=2$, $n=3$)
Let us consider the following three outcome qubit measurement $\N$:
\begin{align*}
\N(1) &=\tfrac{1}{3} (\id + \sigma_1) \\
\N(2) &=\tfrac{1}{3} (\id - \tfrac{1}{2} \sigma_1 + \tfrac{\sqrt{3}}{2} \sigma_2 ) \\
\N(3) &=\tfrac{1}{3} (\id - \tfrac{1}{2} \sigma_1 - \tfrac{\sqrt{3}}{2} \sigma_2 )
\end{align*}
From \eqref{eq:pqracopt_mn}, we obtain $\bar{P}(\M_{1},\N)\simeq 0.695$, $\bar{P}(\M_{2},\N)\simeq 0.696$ and $\bar{P}(\M_{3},\N)\simeq 0.717$.
The bound in the right hand side of \eqref{eq:compatible-nm} is $2/3$, hence the incompatibility of $\M_{k}$ and $\N$ is detected for each $k=1,2,3$.
\end{example}

\begin{example}($d=2$, $m=2$, $n=4$)
Let us consider the following symmetric informationally complete (SIC) qubit measurement $\N$:
\begin{align*}
\N(1) &=\tfrac{1}{4} (\id + \tfrac{1}{\sqrt{3}} (\sigma_1+\sigma_2+\sigma_3)) \\
\N(2) &=\tfrac{1}{4} (\id + \tfrac{1}{\sqrt{3}} (\sigma_1-\sigma_2-\sigma_3)) \\
\N(3) &=\tfrac{1}{4} (\id + \tfrac{1}{\sqrt{3}} (-\sigma_1+\sigma_2-\sigma_3)) \\
\N(4) &=\tfrac{1}{4} (\id + \tfrac{1}{\sqrt{3}} (-\sigma_1-\sigma_2+\sigma_3))
\end{align*}
We obtain $\bar{P}(\M_{k},\N)\simeq 0.6465$ for $k=1,2,3$.
The bound in the right hand side of \eqref{eq:compatible-nm} is $0.625$, hence the incompatibility of $\M_{k}$ and $\N$ is detected for each $k=1,2,3$.
\end{example}

\begin{example}($d=3$, $m=n=2$)
Let us consider two dichotomic qutrit measurements $\M$ and $\N$, defined by
$$
\M(1)=\begin{bmatrix}
1 & 0 & 0 \\
0 & 1 & 0 \\
0 & 0 & 0
\end{bmatrix}
\, , \quad 
\N(1)=\frac{1}{3} \begin{bmatrix}
1 & 1 & 1 \\
1 & 1 & 1 \\
1 & 1 & 1
\end{bmatrix}
$$
and $\M(2)=\id-\M(1)$, $\N(2)=\id-\N(1)$. 
Then $\bar{P}(\M,\N)\simeq 0.9013$ while the bound in the right hand side of \eqref{eq:compatible-nm} is $0.875$.
\end{example}

\section{QRAC bounds for the joint measurability degree and the incompatibility robustness}\label{app:C}

Here, we prove the bounds \eqref{eq:meas_degree_bound}, \eqref{eq:robustness_bound} for the incompatibility robustness and the joint measurability degree defined in \eqref{eq:meas_degree}, \eqref{eq:robustness}.

Suppose $0\leq t\leq\mathsf{j}(\M_1,\M_2)$. Then there exists probabilities $p_1$ and $p_2$ such that the measurements $\{t\M_i+(1-t)p_i\id\}_{i=1,2}$ are compatible. Hence, by Prop.~\ref{prop:compatible-nm},
\begin{equation}
\bar{P}(t\M_1+(1-t)p_1\id,\,t\M_2+(1-t)p_2\id) \leq \frac{1}{2}\left(1+\frac{d}{nm}\right) \,.
\end{equation}
Combining this inequality with \eqref{eq:p1p2}, we obtain
\begin{equation}
\begin{aligned}
& t\left[\bar{P}(\M_1,\M_2) - \frac{1}{2}\left(1+\frac{d}{nm}\right)\right] \leq \\
& \qquad\qquad \leq (1-t) \left[\frac{1}{2}\left(1+\frac{d}{nm}\right) - \frac{n+m}{2nm}\right] \,.
\end{aligned}
\end{equation}
Whenever \eqref{eq:compatible-nm} is satisfied, the differences in both square parentheses are positive. Solving with respect to $t$ thus yields
\begin{equation}
t\leq \frac{d+nm-(m+n)}{2nm\bar{P}(\M_1,\M_2) - (n+m)} \,,
\end{equation}
which implies the bound \eqref{eq:meas_degree_bound}.

Suppose $t\geq\ca{R}(\M_1,\M_2)$ and let $s=1/(1+t)$. Then there exists measurements $\N_1$ and $\N_2$ such that, setting $\M_i' = s\M_i+(1-s)\N_i$, the convex combinations $\M_1'$ and $\M_2'$ are compatible. As before, Prop.~\ref{prop:compatible-nm} thus implies
\begin{equation}
\bar{P}(\M_1',\M_2') \leq \frac{1}{2}\left(1+\frac{d}{nm}\right) \,.
\end{equation}
If $\en$ is any encoding map optimizing the average success probability \eqref{eq:average_mn}, then
\begin{equation}
\begin{aligned}
& s\bar{P}(\M_1,\M_2) = \\
& \qquad = \frac{1}{2nm} \sum_{x=1}^n \sum_{y=1}^m\tr{\en(x,y) (s\M_1(x) + s\M_2(y))} \\
& \qquad \leq \frac{1}{2nm} \sum_{x=1}^n \sum_{y=1}^m\tr{\en(x,y) (\M_1' (x) + \M_2'(y))} \\
& \qquad \leq \bar{P}(\M_1',\M_2') \,.
\end{aligned}
\end{equation}
Therefore,
\begin{equation}
s\leq\frac{d+nm}{2nm\bar{P}(\M_1,\M_2)} \,,
\end{equation}
or, equivalently,
\begin{equation}
t\geq\frac{2nm\bar{P}(\M_1,\M_2)}{d+nm} - 1 \,,
\end{equation}
which implies the bound \eqref{eq:robustness_bound}.

\end{document}